\documentclass[%
 reprint,
 longbibliography,
 amsmath,amssymb,
 aps,
 pra,
]{revtex4-2}

\usepackage[mathlines]{lineno}
\usepackage{braket}
\usepackage{amsmath,amssymb}
\usepackage{amsthm}
\usepackage{graphicx}
\usepackage{latexsym}

\theoremstyle{remark}
\newtheorem{theorem}{Theorem}
\newtheorem{example}{Example}
\newtheorem{lemma}{Lemma}

\newtheorem{remark}{Remark}

\begin{document}

\title{Advance sharing of quantum shares for quantum secrets}


\author{Mamoru Shibata}
\email[]{shibata@it.ict.e.titech.ac.jp}
\homepage[]{https://orcid.org/0000-0002-0154-0514}
\affiliation{Department of Information and Communications Enginnering, Tokyo Institute of Technology, Tokyo 152-8550 Japan} 
\author{Ryutaroh Matsumoto}
\affiliation{Department of Information and Communications Enginnering, Tokyo Institute of Technology, Tokyo 152-8550 Japan} 
\affiliation{Department of Mathematical Sciences, Aalborg University, 9220 Aalborg, Denmark}

\begin{abstract}
Secret sharing is a cryptographic scheme to encode a secret to multiple shares being distributed to participants, so that only qualified sets of participants can restore the original secret from their shares.
When we encode a secret by a secret sharing scheme and distribute shares, sometimes not all participants are accessible, and it is desirable to distribute shares to those participants before a secret information is determined.
Secret sharing schemes for classical secrets have been known to be able to distribute some shares before a given secret.
Lie et al.\ found any pure $(k,2k-1)$-threshold secret sharing for quantum secrets can distribute some shares before a given secret. 
However, it is unknown whether distributing some shares before a given secret is possible with other access structures of secret sharing for quantum secrets.
We propose a quantum secret sharing scheme for quantum secrets that can distribute some shares before a given secret with other access structures. 
\end{abstract}

\maketitle 

\section{Introduction}
To protect important information from destruction or loss, we should not store it in one place, but we should store copies of it across multiple places and media. 
However, if the important information is secret, this strategy clearly increases the risk of information leakage.
A revolutionary method to solve this problem is the secret sharing (SS), which was invented independently by Shamir \cite{SS_Shamir_1979} and Blakley \cite{SS_Blakley_1979} in 1979.
SS is a cryptographic scheme to encode a secret to multiple shares being distributed to participants, 
so that certain sufficiently large sets of participants can restore the secret from their shares.
In quantum information theory, Hillery et al.\ \cite{QSS_Hillery_1999} and Cleve et al.\ \cite{QSS_Cleve_1999} simultaneously presented the quantum secret sharing (QSS) scheme in 1999.
Cleve et al.\ clarified the relationships between QSS and quantum error-correcting codes. 
In that relations, a share of QSS is each qubit of a codeword in a quantum error-correcting code \cite{QSS_Cleve_1999}.
Quantum mechanics extends the capabilities of secret sharing beyond those of classical secret sharing \cite{Gottesman_QSS_2000}. 
QSS is actively studied recently \cite{efficientQSS_Senthoor_2021, concatenatingQSS_Senthoor_2022,Senthoor_2022}.
A set of participants that can restore a secret is called a qualified set, and 
a set of participants that can gain no information about a secret is called a forbidden set.
The set of qualified sets and that of forbidden sets are called an access structure. 
A set of participants that are not qualified set is called an unqualified set.

SS for quantum secrets can be classified into two categories. One is perfect QSS and the other is non-perfect or ramp QSS \cite{rampQSS_Ogawa_2005}.
In a perfect QSS, every unqualified set is a forbidden set. 
A major disadvantage of perfect SS is that the size of each share must be larger than or equal to that of secret \cite{Gottesman_QSS_2000}. 
By tolerating partial information leakage to unqualified sets, the size of shares can be smaller than that of secret.
Such QSS is called ramp QSS or non-perfect QSS.
The ramp QSS was proposed by Ogawa et al.\ \cite{rampQSS_Ogawa_2005}.
In an $(a,k,n)$ ramp QSS, a dealer encodes $k$ qudits of a quantum secret 
into $n$ shares in such a way that any $a$ or more shares can restore the secret 
while any $(a-k)$ or fewer shares has no information about the secret.

Sometimes some participants are inaccessible after the dealer obtains a secret. 
The following situation was considered in \cite{Miyajima_2022}.
In a country, the president suffers from a serious disease and is anxious sbout his sudden death. 
He is afraid that his death makes a national secret inaccessible to anyone if he alone knows about a national secret. 
For this reason, the president wishes to share a national secret to the dignitaries by a secret sharing scheme. 
A national secret is sensitive information and the president needs to hand encoded information of a national secret to the dignitaries. 
The president will obtain a national secret three days later but some dignitaries will make an extended business trip to foreign country from tomorrow.
How can the president share the secret? 
In this situation, it is desirable for the dealer to distribute shares to some participants while the dealer can communicate with them. 
To realize this distribution, the dealer needs to be capable to distribute shares to some participants before a given secret.

We call a distribution of shares to some participants before a given secret ``advance sharing" and a set of shares 
that can be distributed in advance is called ``advance shareable" \cite{Miyajima_2022}. 
A pure state QSS is a QSS such that both secret and whole shares are pure states \cite{QSS_Cleve_1999}. 
A perfect $(a,2a-1)$ threshold quantum secret sharing is a $(a,1,2a-1)$ ramp QSS.
Lie et al.\ \cite{Lie_2020} found that any pure state perfect $(a,2a-1)$ threshold quantum secret sharing can distribute $(a-1)$ shares before a given secret. 
However, it is unknown whether advance sharing is possible with non-threshold or ramp QSS.
We propose a scheme of advance sharing of quantum shares for quantum secret sharing, which 
is able to construct a ramp QSS and non-threshold QSS.

Brun et al.\ \cite{EAQECC_Brun_2006} proposed entanglement-assisted quantum error-correcting codes (EAQECCs).
An EAQECC encodes $k$ information qudits with the help of $c$ maximally entangled pairs.
An $[[n,k;c]]_p$ EAQECC works as follows: 
\begin{enumerate}
    \item Before the qunautum communication begins, a sender and a receiver share some maximally entangled pairs.
    \item The sender encodes $k$ information qudits $\ket{\psi_k}$ together with $\ell=n-c-k$ ancilla qudits and the sender's half of the $c$ entangled pairs into $n$ qudits $\rho_n$.
    \item The sender sends $\rho_n$ to the receiver through a noisy communication channel.
    \item The receiver combines the noisy received qudits with the receiver's half of the $c$ entangled pairs and performs measurements on all $(n+c)$ qudits to distinguish the error.
    \item The receiver performs a recovery operation and restores the $k$ information qudits.
\end{enumerate}
An error whose position is known is called an erasure.
EAQECCs can also correct erasures.

We can construct a QSS capable of advance sharing by distributing $c$ halves of maximally entangled pairs to some participants 
before a given secret, then distributing each qudit of $\rho_n$ to remaining participants after a given secret.
A set of participants can restore the secret by erasure correction procedure of EAQECC. 
In practical use, the access structure of QSS should be clear. 
However, since erasures of receiver's $c$ halves of maximally entangled pairs in EAQECC are not considered, 
it is difficult to clarify the access structures of the QSS considered in this paragraph. 
So, we give a construction of EAQECC from a stabilizer, 
which enables us to analyze the access structure of QSS capable of advance sharing.
By using our proposed construction of an EAQECC, we propose a QSS for quantum secrets 
that can distribute some shares before a given secret. 
Then, we clarify a necessary and sufficient condition on advance-shareable sets in our proposal.

Our proposed QSS can have an access structure that cannot be constructed by the schemes of Lie et al.\ \cite{Lie_2020}.
The schemes of Ogawa et al.\ cannot construct an $(a,k,n)$ ramp QSS 
whose share has dimension $n$ or less \cite{rampQSS_Ogawa_2005}.
We will give an example of advanced sharing for a $(3,2,4)$ ramp QSS 
with 1-qubit shares.
For an $[[n,k;c]]_p$ EAQECC, we usually desire a small value of $c$, and the advantage of a large $c$ was little known.
Our proposed scheme shows the significance of constructing an EAQECC with a large $c$ 
because the size of the advance shareable set increases by larger $c$.

This paper is organized as follows.
In Section \ref{sec:2}, we review stabilizer codes and EAQECCs. 
In Section \ref{sec:3}, we give a construction of EAQECC from a stabilizer, which later enables us to analyze the access structure of resulting QSS.
Then, we propose a scheme of advance sharing for QSS by EAQECCs, and we clarify necessary and sufficient conditions of advance shareable sets
of our proposal. 
We propose a sufficient condition of advance shareable set that can be verified without computing dimensions of linear spaces.
The conclusions follow in Section \ref{sec:four}.

\section{Preliminaries}\label{sec:2}
In this section, we review stabilizer codes and EAQECCs.
Throughout this paper, we suppose that $p$ is a prime number.
\subsection{Stabilizer codes}
Let $\{\ket{i}\mid i=0,\dots,p-1\}$ be an orthonormal basis for $p$-dimensional Hilbert space $\mathbb{C}^p$.
Let $\omega$ be a complex number such that is $\omega^p=1$ and $\omega^1,\omega^2,\dots,\omega^{p-1}$ are different.
We define two unitary matrices $X_p,Z_\omega$ 
that change $\ket{i}$ as 
$X_p \ket{i} = \ket{i+1\mod p}$ 
and $Z_\omega\ket{i} = \omega^i\ket{i}$ for $i=0,\dots,p-1$.
Consider the set $E_n = \{\omega^iX_p^{a_1}Z_\omega^{b_1}\otimes \dots\otimes X_p^{a_n}Z_\omega^{b_n}\mid i,a_j,b_j\in\{0,\dots,p-1\}\text{ for }j = 1,\dots,n\}$.
$E_n$ is a non-commutative finite group with matrix multiplication as its group operation. 
Denote by $\mathbb{F}_p$ the finite field with $p$ elements. 
For $\vec{a}=(a_1,\dots ,a_n),\vec{b} =(b_1,\dots ,b_n)\in\mathbb{F}_p^n$ define 
$X_p(\vec{a}) = X_p^{a_1}\otimes\dots\otimes X_p^{a_n}$ 
and $Z_\omega(\vec{b})=Z_\omega^{b_1}\otimes\dots\otimes Z_\omega^{b_n}$.
For two vectors $(\vec{a}\mid \vec{b}),(\vec{c}\mid \vec{d})\in \mathbb{F}_p^{2n}$, the sympectic inner product is defined by
\begin{equation*}
    \langle(\vec{a}\mid \vec{b}),(\vec{c}\mid \vec{d})\rangle_s=\langle\vec{a},\vec{d}\rangle_E-\langle\vec{b},\vec{c}\rangle_E, 
\end{equation*}
where $\langle\cdot \mid \cdot\rangle_E$ is the Euclidean inner product.
We define the weight of $\omega^{i}X_p(\vec{a})Z_\omega(\vec{b})\in E_n$ as $w(\omega^i X_p(\vec{a})Z_\omega(\vec{b}))=\sharp\{i\mid (a_i,b_i)\neq0\}$.
We call a commutative subgroup of $E_n$ as a stabilizer. 
Let $S$ be a stabilizer of $E_n$. 
Let $S' = \{M\in E_n\mid MN = NM  \text{ for } \forall N\in S\}$, and let $\overline{S}$ be the commutative subgrous of $E_n$ generated by $\omega I_p^{\otimes n}$ and $S$.
Here $I_p$ is the identity matrix on $\mathbb{C}^p$. 
We define the minimum distance of a stabilizer $S$ by $d(S)=\min\{w(M)\mid M\in S'\backslash\overline{S}\}$.

Suppose that eigenspaces of a stabilizer $S$ has dimension $p^k$. 
An $[[n,k]]_p$ quantum stabilizer code $Q(S)$ encoding $k$ qudits into $n$ qudits can be defined as 
a simultaneous eigenspace of all elements of $S$.
Sometimes we will write $[[n,k,d]]_p$ stabilizer code to indicate that the distance of the code is $d$. 
An $[[n,k,d]]_p$ stabilizer code is capable of correcting less than $d$ erasures.

Now, we explain a way to describe a stabilizer $S$ by finite fields. 
For an $(n-k)$-dimensional $\mathbb{F}_p$-linear subspace $C$ of $\mathbb{F}_p^{2n}$, 
we define $C^{\perp} = \{\vec{a}\in\mathbb{F}^{2n}_p\mid \forall \vec{b}\in C,\langle\vec{a},\vec{b}\rangle_s=0\}$. 
We define $M(\vec{a}|\vec{b})$ as $M(\vec{a}|\vec{b}) = X_p(\vec{a})Z_\omega(\vec{b})\in {E_n}$ with $\vec{a}, \vec{b}\in \mathbb{F}_p^n$. 
We define a mapping $f(\omega^i M(\vec{a}|\vec{b}))$ from ${E_n}$ to $\mathbb{F}_p^{2n}$ by $f(\omega^i M(\vec{a}|\vec{b})) = (\vec{a}|\vec{b})$.
For a stabilizer $S$, $f(S)$ is an $\mathbb{F}_p$-linear space.
We define a check matrix of a stabilizer $S$ as a matrix $H_{S}=[H_X\mid H_Z]$ whose row space is $f(S)$.

\begin{example}\label{ex_stabilizer}
    We define a stabilizer generator $\{M_1,M_2\}$ as follows: 
    \begin{eqnarray*}
        \begin{array}{ccclclclclc}
            M_1 &=& X_2&\otimes&X_2&\otimes&X_2&\otimes&X_2,\\
            M_2 &=& Z_{-1}&\otimes&Z_{-1}&\otimes&Z_{-1}&\otimes&Z_{-1}.
        \end{array}
    \end{eqnarray*}
    We define a basis of the simultaneous $+1$ eigenspace $Q$ of this stabilizer $\{\ket{00_L},\ket{01_L},\ket{10_L},\ket{11_L}\}$ as follows:
    \begin{eqnarray*}
        \begin{array}{ccc}
            \ket{00_L} &=& \frac{1}{\sqrt{2}}(\ket{0000}+\ket{1111}),\\
            \ket{01_L} &=& \frac{1}{\sqrt{2}}(\ket{0011}+\ket{1100}),\\
            \ket{10_L} &=& \frac{1}{\sqrt{2}}(\ket{0101}+\ket{1010}),\\
            \ket{11_L} &=& \frac{1}{\sqrt{2}}(\ket{0110}+\ket{1001})
        \end{array}
    \end{eqnarray*}
    The dimension of $Q$ is 4 and the minimum distance of this stabilizer is 2.
    Thus, $Q$ is a $[[4,2,2]]_2$ stabilizer code.
    A check matrix of this stabilizer can be written as follows:
    \begin{equation*}
        H=
        \left[
        \begin{array}{cccc|cccc}
            1&1&1&1 &0&0&0&0\\
            0&0&0&0 &1&1&1&1
        \end{array}
        \right].
    \end{equation*}
    This stabilizer code is capable of correcting 1 erasure.
\end{example}

\subsection{EAQECC}
We review $p$-ary EAQECC \cite{NonbinaryEAQECC_Luo_2017,EAQECC_Brun_2006}.
Suppose that a sender and a receiver share $c$ pairs of maximally entangled states.
For an arbitrary non-abelian subgroup $G\subset E_n$, 
there exist a set of generators $\{\overline{Z}_1,\overline{Z}_2,\dots,\overline{Z}_{c+\ell},\overline{X}_1,\dots,\overline{X}_{c}\}$ 
for $G$ where $c \geq 1,k\geq 0, \ell=n-c-k$ with the following commutative relations:
\begin{equation*}
    \begin{array}{ll}
        \lbrack\overline{X}_i,\overline{X}_j\rbrack = 0, & \forall 1\leq i,j\leq c,\\
        \lbrack\overline{Z}_i,\overline{Z}_j\rbrack = 0, & \forall 1\leq i,j\leq c+\ell,\\
        \lbrack\overline{X}_i,\overline{Z}_j\rbrack = 0, & \forall i \neq j, 1\leq i \leq c,1\leq j \leq c+\ell,\\
        \lbrack\overline{X}_i,\overline{Z}_i\rbrack \neq 0, & \forall 1\leq i \leq c.
    \end{array}
\end{equation*}
Here $\lbrack A,B\rbrack$ is the commutator, $\lbrack A,B\rbrack = AB-BA$ for $A,B\in E_n$.
Let $\mu_i$ be an integer such that $\overline{X}_i\overline{Z}_i=\omega^{-\mu_i}\overline{Z}_i\overline{X}_i$.
We define $X_{(i)},Z_{(j)}$ for $i=1,2,\dots ,c$ and $j=1,2,\dots,c+\ell$ as:
\begin{equation*}
    \begin{array}{ll}
        X_{(i)} = I^{\otimes{i-1}}\otimes X_p^{\mu_i}\otimes I^{\otimes{n-i}},\\
        Z_{(j)} = I^{\otimes{j-1}}\otimes Z_\omega\otimes I^{\otimes{n-j}}.
    \end{array}
\end{equation*}
We define a subgroup $B_G$ of $E_n$ generated by $\{Z_{(1)},\dots,Z_{(c+\ell)},X_{(1)},\dots,X_{(c)}\}$. 

Since the groups $B_G$ and $G$ are isomorphic as groups,
 we can relate $B_G$ to $G$ by following lemma \cite{NonbinaryEAQECC_Luo_2017}:
\begin{lemma}\label{isomorphic}
    If $B_G$ is defined as above, then there exists a unitary $U$ such that for all $b\in B_G$ there exists an $g\in G$ such that $b=U g U^\dagger$ up to an overall phase.\hfill $\Box$
\end{lemma}
We define $X'_{(i)}, Z'_{(j)},$ as:
\begin{equation*}
    \begin{array}{ll}
        X'_{(i)} = X_{(i)}\otimes I^{\otimes i-1}\otimes X_p^{\mu_i}\otimes I^{\otimes c-i}, & 1\leq i\leq c,\\
        Z'_{(j)} = Z_{(j)}\otimes I^{\otimes{j-1}}\otimes Z^{-1}_\omega\otimes I^{\otimes{c-j}}, & 1\leq j\leq c,\\
        Z'_{(j)} = Z_{(j)}\otimes I^{\otimes c}, & c< j \leq c+\ell.
    \end{array}
\end{equation*}
Let $B_G'$ be a group generated by $\{Z'_{(1)},\dots,Z'_{(c+\ell)},X'_{(1)},\dots,X'_{(c)}\}$.
Then, $B_G'$ is a stabilizer of ${E_{n+c}}$ because $B_G'$ is a commutative subgroup of ${E_{n+c}}$.
For an arbitrary $k$ qudits $\ket{\psi_k}$, 
the codeword $\ket{\Psi}$ of a stabilizer code $Q(B_G')$ can be written as follows:
\begin{equation*}
        \ket{\Psi} = 
    \sum_{(i_1,\ldots,i_c)\in\mathbb{F}_p^c}
    \frac{1}{\sqrt{p^c}}
    \ket{i_1}\dots\ket{i_c}
    \ket{0}^{\otimes \ell}
    \ket{\psi_k}
    \ket{i_1}\dots\ket{i_c}
\end{equation*}
where the pairs of $j$th and $(n+j)$th qudits $(j=1,2,\dots,c)$ of $\ket{\Psi}$ form maximally entangled pairs. 
The $(n+1)$th through $(n+c)$th qudits of $\ket{\Psi}$ are the receiver's $c$ halves of maximally entangled pairs.
We define $\overline{Z}'_{i},\overline{X}'_{i}$ as
\begin{equation*}
    \begin{array}{ll}
        \overline{X}'_{i} = \overline{X}_{i}\otimes I^{\otimes i-1}\otimes X_p^{\mu_i}\otimes I^{\otimes c-i}, & 1\leq i\leq c,\\
        \overline{Z}'_{j} = \overline{Z}_{j}\otimes I^{\otimes{j-1}}\otimes Z^{-1}_\omega\otimes I^{\otimes{c-j}}, & 1\leq j\leq c,\\
        \overline{Z}'_{j} = \overline{Z}_{j}\otimes I^{\otimes c}, & c< j \leq c+\ell.
    \end{array}
\end{equation*}
Let $G'$ be a group generated by $\{\overline{Z'}_{1},\dots,\overline{Z'}_{c+\ell},\overline{X'}_{1},\dots,\overline{X'}_{c}\}$.
Then, $G'$ is a stabilizer of ${E_{n+c}}$ because $G'$ is a commutative subgroup of ${E_{n+c}}$.
We define a stabilizer code $Q(G')$. 
From Lemma 1, 
a code space $Q(G')$ is given by 
\begin{equation*}
    Q(G') = \{(U\otimes I^{\otimes c})\ket{\Psi}\mid 
    \ket{\Psi}\in Q(B_G')
    \}.
\end{equation*} 
The sender performs the encoding operation $U$ on information qudits $\ket{\psi_k}$, the sender's halves of the entangled pair, and $\ell = n-k-c$ ancilla qudits.
The sender then sends $n$ qudits through an noisy channel to the receiver. 
The receiver combines the received $n$ qudits and the receiver's $c$ halves of the entangled pair.
The receiver correct the erasures in the resulting $(n+c)$ qudits by the stabilizer code $Q(G')$, 
and decode the information qudits $\ket{\psi_k}$ by performing $(U\otimes I^{\otimes c})^{-1}$.

Then $(G,G',U)$ is an $[[n,k,d;c]]_p$ EAQECC 
that employs $c$ maximally entangled pairs and $\ell = n-k-c$ ancilla qudits 
to encode $k$ infromation qudits.
The erasure correcting ability of $(G,G',U)$, including receiver's halves of maximally entangled pairs, 
is the same as a stabilizer code $Q(G')$.

\begin{example}\label{ex_EAQECC}
    We define a non-abelian subgroup $G\subset E_3$ generator $\{\overline{X}_1,\overline{Z}_1\}$ as follows: 
    \begin{eqnarray*}
        \begin{array}{ccclclclc}
            \overline{X}_1 &=& X_2&\otimes&X_2&\otimes&X_2,\\
            \overline{Z}_1 &=& Z_{-1}&\otimes&Z_{-1}&\otimes&Z_{-1}.
        \end{array}
    \end{eqnarray*}
    Then, generators of $B_G$ can be written as follows:
    \begin{eqnarray*}
        \begin{array}{ccclclclc}
            {X}_{(1)} &=& X_2&\otimes&I_2&\otimes&I_2,\\
            {Z}_{(1)} &=& Z_{-1}&\otimes&I_2&\otimes&I_2.
        \end{array}        
    \end{eqnarray*}
    From Lemma 1, we can find a unitary matrix $U$ as follows:
    \begin{eqnarray*}
        U = \ket{000}\bra{000} + \ket{001}\bra{111} + \ket{010}\bra{110} + \ket{011}\bra{010}\\ 
        + \ket{100}\bra{011} + \ket{101}\bra{100} + \ket{110}\bra{110} + \ket{111}\bra{001}.
    \end{eqnarray*}
    Generators $\{{X'}_{(1)},{Z'}_{(1)}\}$ of $B'_G$ can be written as follows:
    \begin{eqnarray*}
        \begin{array}{ccclclclclc}
            {X'}_{(1)} &=& X_2&\otimes&I_2&\otimes&I_2&\otimes&X_2,\\
            {Z'}_{(1)} &=& Z_{-1}&\otimes&I_2&\otimes&I_2&\otimes&Z_{-1}.
        \end{array}        
    \end{eqnarray*}
    $B_G'$ is a stabilizer of $E_{3+1}$. For an arbitrary $2$ qubits $\ket{\psi_2}$, the codeword $\ket{\Psi}$ of a stabilizer code $Q(B_G')$ can be written as follows:
    \begin{equation*}
            \ket{\Psi} = 
        \frac{1}{\sqrt{2}}
        \Bigl(\ket{0}\otimes\ket{\psi_2}\otimes\ket{0} + \ket{1}\otimes\ket{\psi_2}\otimes\ket{1}\Bigr).
    \end{equation*}

    Then, generators $\{\overline{X'}_{1},\overline{Z'}_{1}\}$ of $G'$ can be written as follows:
    \begin{eqnarray*}
        \begin{array}{ccclclclclc}
            \overline{X'}_1 &=& X_2&\otimes&X_2&\otimes&X_2&\otimes&X_2,\\
            \overline{Z'}_1 &=& Z_{-1}&\otimes&Z_{-1}&\otimes&Z_{-1}&\otimes&Z_{-1}.
        \end{array}
    \end{eqnarray*}
    $G'$ is a stabilizer of $E_{3+1}$. The code space $Q(G')$ is given by
    \begin{equation*}
        Q(G') = \{(U\otimes I_2)\ket{\Psi}\mid 
        \ket{\Psi}\in Q(B_G')
        \}.
    \end{equation*}
    Then $(G, G',U)$ is an $[[3,2,2;1]]_p$ EAQECC 
    that employs $1$ maximally entangled pairs and $0$ ancilla qudits 
    to encode $2$ infromation qudits.
\end{example}
\subsection{Stabilizer-based QSS}
We review a stabilizer-based QSS \cite{QSS_Cleve_1999}. 
It is accomplished by the following steps:
\begin{enumerate}
    \item A dealer encodes a quantum secret by a stabilizer code.
    \item The dealer distributes each qudit of that codeword to a participant.
\end{enumerate}
There are some procedures to restore the secret for stabilizer-based QSS \cite{Unitary_reconstruction_Matsumoto_2017}.
One of the simplest procedure is to use erasure correction of the stabilizer code \cite{QSS_Cleve_1999}.
The access structure of a stabilizer-based QSS depend on the stabilizer code. 
Stabilizer-based QSS can construct a ramp QSS from $[[n,k]]_p$ stabilizer codes with $k\geq 2$.
\section{QSS constructed from EAQECC}\label{sec:3}
In this section, we propose a scheme of advance sharing for QSS by EAQECC.
First, we construct an EAQECC from a stabilizer 
to clarify the access structure of our proposal. 
Second, we propose a construction of QSS from EAQECC. 
Third, we clarify necessary and sufficient conditions for an index set $J\subset\{1,2,\dots,n\}$ 
to be an advance shareable set.
Finaly, we present a sufficient condition of advance shareable set for an index set $J\subset\{1,2,\dots,n\}$.
\subsection{EAQECC constructed from a stabilizer}
For a stabilizer $S$, 
we introduce a construction of an EAQECC 
that has the same erasure correcting ability, including receiver's halves of maximally entangled pairs, 
as its stabilizer code $Q(S)$.

Let $J\subset\{1,2,\dots,n\}$ be an index set.
For a check matrix $H_S$ of a stabilizer and an index set $J$, we define conditions 1 and 2 as follows:
\begin{enumerate}
    \item For $j\in J$, $j$th column of $H_S$ has 1 at only $j$th row and its other rows are 0.
    \item For $j\in J$, $(n+j)$th column of $H_S$ has 1 at only $(n+j)$th row and its other rows are 0.
\end{enumerate}

\begin{lemma}\label{matrix}
    If a check matrix $H_S$ of a stabilizer $S$ satisfy the conditions 1 and 2, 
    there exists a check matrix $H'_S$ of $S$ that is written as follows:
    \begin{equation*}
        H'_{S}=
        \left[
        \begin{array}{cccc|cccc}
            h_{1,1}&\cdots&0&\cdots&\cdots&0&\cdots&h_{1,2n}\\
            \vdots&&\vdots&&&\vdots&&\vdots\\
            h_{i,1}&\cdots&\mu_i&\cdots&\cdots&0&\cdots&h_{i,2n}\\
            \vdots&&\vdots&&&\vdots&&\vdots\\
            h_{i+|J|,1}&\cdots&0&\cdots&\cdots&-1&\cdots&h_{i+|J|,2n}\\
            \vdots&&\vdots&&&\vdots&&\vdots\\
            h_{n-k,1}&\cdots&0&\cdots&\cdots&0&\cdots&h_{n-k,2n}
        \end{array}
        \right]
    \end{equation*}
    where $\mu_i = \sum_{j=1, j\neq i}^{n}h_{i,j}h_{i+|J|,n+j} - \sum_{j=1, j\neq i}^{n}h_{i+|J|,j}h_{i,n+j}$.
    Since $S$ is an abelian subgroup of $E_n$, we have $\mu_i \neq 0$.\hfill $\Box$
\end{lemma}

Let $S$ be a stabilizer of $E_n$.
Let $\overline{J}=\{1,\dots,n\}\backslash J$.
\begin{lemma}\label{condition_ok}
    There are a unitary matrix $U_J$ and non-abelian subgroup $S^J$ of $E_{n-|J|}$
    such that 
    $(S^J, S, U_J)$ is an $[[n-|J|,k,d;|J|]]_p$ EAQECC that has the same erasure correcting ability as its stabilizer code $Q(S)$ 
    if $J$ and a check matrix $H_S$ of $S$ satisfy the conditions 1 and 2.
\end{lemma}
\begin{proof}
    From Lemma \ref{matrix}, if a check matrix $H_S$ of $S$ satisfy the conditions 1 and 2,
    there exists a check matrix $H'_S$ of $S$ that is written as follows:
    \begin{equation*}
        H'_{S}=
        \left[
        \begin{array}{cccc|cccc}
            h_{1,1}&\cdots&0&\cdots&\cdots&0&\cdots&h_{1,2n}\\
            \vdots&&\vdots&&&\vdots&&\vdots\\
            h_{i,1}&\cdots&\mu_i&\cdots&\cdots&0&\cdots&h_{i,2n}\\
            \vdots&&\vdots&&&\vdots&&\vdots\\
            h_{i+|J|,1}&\cdots&0&\cdots&\cdots&-1&\cdots&h_{i+|J|,2n}\\
            \vdots&&\vdots&&&\vdots&&\vdots\\
            h_{n-k,1}&\cdots&0&\cdots&\cdots&0&\cdots&h_{n-k,2n}
        \end{array}
        \right].
    \end{equation*}
    Then, we can define generators $\{G_1,\ldots,G_{n-k}\}$ of $S$ as follows:
    \begin{eqnarray*}
        \begin{array}{cc}
            G_i = \bigotimes^{n}_{j=1} X_p^{h_{i,j}}Z_\omega^{h_{i,n+j}}, & i = 1,\ldots,n-k            
        \end{array}
    \end{eqnarray*}
    where $h_{i,j}$ is the $(i,j)$ component of $H'_S$.
    We define $\{G^J_1,\ldots,G^J_{n-k}\}$ as follows:
    \begin{equation*}
        G^J_i = \bigotimes^{n}_{\substack{j=1,\\j \notin J}} X_p^{h_{i,j}}Z_\omega^{h_{i,n+j}},  i = 1,\ldots,n-k.
    \end{equation*}
    Let $S^J$ be a subgroup of ${E_{n-|J|}}$ generated by $\{G^J_1,\ldots,G^J_{n-k}\}$.
    We define $\{x_1,x_2,\ldots,x_{|J|}\}\subset \overline{J}$ and 
    $\{z_{2|J|+1},z_{2|J|+2},\ldots,z_{n-k}\}\subset \overline{J}$ such that all of $x_i, z_i$ are 
    different from each other.
    We define $g_{j,l}$ as follows:
    \begin{equation*}
        \begin{array}{cc}
            g_{j,j} = g_{j,x_j} = X_p^{\mu_j}, & j \in J,\\
            g_{j+|J|,j} = Z_\omega, & j \in J,\\
            g_{j+|J|,x_j} = Z^{-1}_\omega, & j \in J,\\
            g_{i,z_i} = Z_\omega, & i \in\{2|J|+1,\ldots,n-k\},\\
            g_{i,j} = I_p, & \text{otherwise}.
        \end{array}
    \end{equation*}
    We define $\{{G'}^J_1,\ldots,{G'}^J_{n-k}\}$ and $\{{G'}_1,\ldots,{G'}_{n-k}\}$ as follows:
    \begin{eqnarray*}
        {G'}^J_j = \bigotimes^{n}_{\substack{l=1,\\l \notin J}} g_{j,l},\\
        {G'}_j = \bigotimes^{n}_{\substack{l=1}} g_{j,l}.
    \end{eqnarray*}
    Let $B^J_S$ be a subgroup of ${E_{n-|J|}}$ generated by $\{{G'}^J_1,\ldots,{G'}^J_{n-k}\}$.
    Considering the mapping $G^J_j \mapsto {G'}^J_j$, we see that $S^J$ and $B^J_S$ are isomorphic.
    Since the groups $B^J_S$ and $S^J$ are isomorphic as groups, 
    there exists a unitary $U_J$ such that 
    for all $b\in B^J_S$ there exists an $g\in S^J$ such that $b=U_JgU_J^\dagger$ up to overall phase.

    Let $B_S$ be a subgroup of ${E_{n}}$ generated by $\{{G'}_1,\ldots,{G'}_{n-k}\}$.
    Let $Q(B_S)$ be a stabilizer code of $B_S$. 
    Since $S^J$ and $B^J_S$ are isomorphic and we have $X_p^{h_{i,j}}Z_\omega^{h_i,n+j} = g_{i,j}$ for $j\in J$, $i = 1,\ldots,n-k$, the groups $S$ and $B_S$ are isomorphic.
    Then, we can construct an $[[n-|J|,k,d;|J|]]_p$ EAQECC $(S^J, S ,U_J)$ by the procedure in Section \ref{sec:2}.
    In the decoding procedure of $(S^J, S ,U_J)$, we correct the erasures by the stabilizer code $Q(S)$.
    So, $(S^J, S, U_J)$ has the same erasure correcting ability as $Q(S)$.
\end{proof}

\begin{example}\label{EAQECC_from_stabilizer}
    Let $S$ be a stabilizer defined in Example \ref{ex_stabilizer}.
    Let $J = \{4\}$ be an index set.
    Here, a check matrix $H_S$ of $S$, which is defined in Example 1, satisfies the conditions 1 and 2.
    We define $S^J$ and its generator $\{G^J_1,G^J_2\}$ as follows:
    \begin{eqnarray*}
        \begin{array}{ccclclclc}
            G^J_1 &=& X_2&\otimes&X_2&\otimes&X_2,\\
            G^J_2 &=& Z_{-1}&\otimes&Z_2&\otimes&Z_2.
        \end{array}        
    \end{eqnarray*}
    Then, $S^J$ is the same as $G$ in Example 2.
    Therefore, $(S^J,S,U_J)$ is a $[[3,2,2;1]]_p$ EAQECC.
    In the decoding procedure of $(S^J, S ,U_J)$, we correct the erasures by the $[[4,2,2]]_p$ stabilizer code $Q(S)$.
\end{example}

\subsection{Our proposed encoding method for QSS}
Let $S$ be a stabilizer of $E_n$.
Let $J$ be an index set that satisfies the conditions 1 and 2 with a check matrix $H_S$ of $S$.
From Lemma \ref{condition_ok}, we can define $S^J$ and $U^J$ such that $(S^J,S,U_J)$ is an $[[n-|J|,k,d;|J|]]_p$ EAQECC
that has the same erasure correcting ability as the stabilizer code $Q(S)$.
We propose a QSS for quantum secrets with an index set $J$ being advance shareable as following:
\begin{enumerate}
    \item A dealer prepares $|J|$ pairs of maximally entangled states and distributes $|J|$ halves of the maximally entangled states to participants in $J$.
    \item The dealer encodes a $k$-qudit quantum secret $\ket{\psi_k}$ into $n-|J|$ qudits of a codeword of EAQECC $(S^J,S,U_J)$. 
    \item The dealer distributes each qudit of the encoded state to the remaining participants.
\end{enumerate}
A qualified set of participants can get a codeword of $Q(S)$ with erasures by attaching arbitrary qudits as the missing shares to available shares. 
Then, they can restore the secret by the erasure correcting of the stabilizer code $Q(S)$. 
In our proposed QSS, 
shares of an index set $J$ can be distributed to some participants before a given secret.
Our proposed QSS from $S$ and $J$ has the same access structure of the stabilizer-based QSS constructed from $S$.
Since $[[n,k,d]]_p$ stabilizer codes can correct less than $d$ erasure, 
the set of $n+1-d$ or more shares is a qualified set of our proposal.
From \cite[Proposition 3]{rampQSS_Ogawa_2005}, the set of less than $d$ shares is a forbidden set of our proposal.

\begin{remark}
    Our poposed QSS is constructed by using an $[[n-|J|,k,d;|J|]]_p$ EAQECC.
    So, the size of advance shareable set $|J|$ is the number of maximally entangled pairs of EAQECC.
    Therefore, our proposal shows the significance of constructing an EAQECC with a large number of maximally entangled pairs.
\end{remark}

\begin{example}
    Let $S$ be a stabilizer defined in Example 1.
    We define an index set $J=\{4\}$ as an advance shareable set.
    Let $(S^J,S,U_J)$ be a $[[3,2,2;1]]_p$ EAQECC defined in Example 3.
    Our proposed encoding with an index set $J$ being advance shareable is as follows:
    \begin{enumerate}
        \item A dealer prepares a maximally entangled pair and distribute a halve of the maximally entangled state to 4th participant.
        \item The dealer encodes a $2$-qubit quantum secret $\ket{\psi_2}$ into $3$ qudits of a codeword of the EAQECC. 
        \item The dealer distributes each qubit of the encoded state to the remaining participants.
    \end{enumerate}
    The shares are codeword of $Q(S)$ that is capable of correcting 1 erasure, 
    so 3 or more participants can restore the secret.
    This QSS encodes $2$ qubits of a quantum secret into $4$ shares in such a way 
    that any $3$ or more shares can restore the secret 
    while any single share has no information about the secret.
    So, this is a $(3,2,4)$ ramp QSS. 
    This ramp QSS cannot be constructed by the scheme of Lie et al.\ \cite{Lie_2020}.
    In addition, since the dimension of a share is 2 and the number of participants is 4, 
    this ramp QSS cannot be constructed by the scheme of Ogawa et al.\ \cite{rampQSS_Ogawa_2005}.
\end{example}

\subsection{Necessary and sufficient condition of advance shareable sets}
Shortening in this paper refers to making a new linear code $C'\subset \mathbb{F}^{2n-2}_p$ from a linear code $C\subset \mathbb{F}_p^{2n}$ 
by selecting vectors in $C$ where the $i$th and the $(n+i)$th components $(1\leq i\leq n)$ are both zero and then eliminating the $i$th and the $(n+i)$th components of the selected vectors.
Let $C_{(s)}^{(J)}$ be the code obtained by shortening the linear code $C$ for the element corresponding to the index set $J\subset\{1,\dots,n\}$.

We clarify a necessary and sufficient condition that a set of shares $J$ is an advance shareable in our proposal.
\begin{theorem}\label{junbi}
    Let $S$ be a stabilizer of ${E_{n}}$. 
    An index set $J\subset\{1,\dots,n\}$ and a check matrix $H_S$ satisfy the conditions 1 and 2 
    if and only if the equation
    \begin{equation}
        \dim{f(S)_{(s)}^{(J)}}=\dim{f(S)}-2|J|
    \end{equation}
    holds.
\end{theorem}
\begin{proof}
    For ease of presentation, without loss of generality we may assume ${J}=\{1,\dots, |{J}|\}$ and $\overline{J}=\{|{J}|+1,\dots,n\}$, by reordering indicies. 
    First, we prove $\dim{f(S)_{(s)}^{(J)}}=\dim{f(S)}-2|J|$ if $J$ and $H_S$ satisfy the conditions 1 and 2.
    From Lemma \ref{matrix}, the check matrix of stabilizer $S$ is written as follows:
    \begin{equation*}
        H_{S}=
        \left[
        \begin{array}{cc|cc}
            D_{|J|} &  A  &  0  &  B'  \\
            0&  A' &-I_{|J|}&B\\
            0&E&0&F
        \end{array}
        \right],
    \end{equation*}
    where $A,B,A',B'$ are $|J|\times(n-|J|)$ matrices, $E,F$ are $(n-k-2|J|)\times(n-|J|)$ matrices
    and $D_{|J|}$ is a diagonal matrix whose $i$th diagonal components are $\mu_i$ that is defined in Lemma \ref{matrix}.
    Since the row space of $H_S$ is $f(S)$, we obtain $\dim{f(S)}^{(J)}_{(s)}= \dim{f(S)}-2|J|$.

    Second, we prove that there exist $H_S$ satisfy the conditions 1 and 2 
    if $\dim{f(S)_{(s)}^{(J)}}=\dim{f(S)}-2|J|$.
    When the dimension of $f(S)$ is reduced by 2 by shortening for $j$th and $(n+j)$th columns, 
    there is a check matrix $H_S$ that can be written as follows \cite{Ueno_2022}:
    \begin{equation*}
        H_{S}=
        \left[
        \begin{array}{cccc|cccc}
            h_{1,1}&\cdots&0&\cdots&\cdots&0&\cdots&h_{1,2n}\\
            \vdots&&\vdots&&&\vdots&&\vdots\\
            h_{j,1}&\cdots&1&\cdots&\cdots&0&\cdots&h_{j,2n}\\
            \vdots&&\vdots&&&\vdots&&\vdots\\
            h_{j+|J|,1}&\cdots&0&\cdots&\cdots&1&\cdots&h_{j+|J|,2n}\\
            \vdots&&\vdots&&&\vdots&&\vdots\\
            h_{n-k,1}&\cdots&0&\cdots&\cdots&0&\cdots&h_{n-k,2n}
        \end{array}
        \right].
    \end{equation*}
    For all of $j\in J$, the dimension of $f(S)$ is reduced by 2 by shortening for $j$th and $(n+j)$th columns.
    Therefore, there exist $H_S$ satisfy the conditions 1 and 2.
\end{proof}

\begin{example}
    Let $S$ be a stabilizer defined in Example 1.
    Let $H_S$ be a check matrix of $S$ defined in Example 1.
    Here, $H_S$ satisfy the conditions 1 and 2.
    Let $J = \{4\}$ be an index set.
    Since we have $f(S)_{(s)}^{(J)} = \{\vec{0}\}$, we have $\dim{f(S)_{(s)}^{(J)}}=0$.
    Therefore, we have $\dim{f(S)_{(s)}^{(J)}}=\dim{f(S)}-2|J|$.
\end{example}

\subsection{Suffcient condition of advance shareable sets}
We present a sufficient condition that a set of shares $J$ is advance shareable in our propsal.
Let $S$ be a stabilizer of $E_n$.
Let $Q(S)$ be a $[[n,k]]_p$ stabilizer code.
Let $J$ be an index set. 
Puncturing in this paper refers to making a new linear code $C'\subset\mathbb{F}_p^{2n-2}$ from a linear code $C\subset\mathbb{F}_p^{2n}$ by eliminating the $i$th and the $(n+i)$th comoponents $(1\leq i\leq n)$ of all vectors in $C$.
Let $C_{(p)}^{(J)}$ be the code obtained by puncturing the linear code $C$ for the element corresponding to the index set $J\subset\{1,\dots,n\}$. 
We define the sympectic weight of $(\vec{a}|\vec{b})\in C$ as $w_s(\vec{a}|\vec{b}) = \sharp\{i\mid (a_i,b_i)\neq0\}$.
We define the minimum weight of $C$ as $d_{min}(C) = \min\{w_s(\vec{a}|\vec{b})\mid (\vec{a}|\vec{b})\in C, (\vec{a}|\vec{b})\neq \vec{0}\}$.

The following sufficient condition can be verified without computing dimensions of linear spaces.
\begin{theorem}
    \begin{equation*}
        \dim{f(S)_{(s)}^{(J)}}=\dim{f(S)}-2|J|
    \end{equation*}
    if 
    \begin{equation}
        |J|<d_{min}(f(S)^\perp)
    \end{equation}
    holds.
\end{theorem}

\begin{proof}
    According to Lemma 1.1 of the reference \cite{Ueno_2022}, for $[[n,k]]_p$ stabilizer code $Q(S)$, 
    if $|J|< d_{min}(f(S)^\perp)$ holds, then we have $\dim{f(S)^\bot} = \dim{(f(S)^\bot)_{(p)}^{(J)}}$.
    We have $f(S)_{(s)}^{(J)} = ((f(S)^{\bot})_{(p)}^{(J)})^{\bot}$.
    Then, 
    \begin{eqnarray*}
        \dim{f(S)_{(s)}^{(J)}} &=& \dim{((f(S)^{\bot})_{(p)}^{(J)})^{\bot}}\\
        &=& 2n-2|J|-\dim{(f(S)^\bot)_{(p)}^{(J)}}\\
        &=& 2n-2|J|-\dim{f(S)^\bot}\\
        &=& \dim{f(S)}-2|J|.
    \end{eqnarray*}
\end{proof}
So, an index set $J$ is advance shareable set in QSS based on a $[[n,k]]_p$ stabilizer code $Q(S)$ if $|J|<d_{min}(f(S)^\perp)$ holds.

\begin{example}
    Let $S$ be a stabilizer defined in Example 1.
    The following set is a basis of $f(S)$:
    \begin{equation*}
        \left\{
            \begin{array}{c}
                (1111|0000),\\
                (0000|1111)
            \end{array}
        \right\}.
    \end{equation*}
    Then the following set is a basis of $f(S)^\perp$:
    \begin{equation*}
        \left\{
            \begin{array}{c}
                (0011|0000),\\
                (0101|0000),\\
                (1001|0000),\\
                (0000|0011),\\
                (0000|0101),\\
                (0000|1001)
            \end{array}
        \right\}.
    \end{equation*}
    We have following identity:
    \begin{equation*}
        d_{min}(f(S)^\perp)=2.
    \end{equation*}
    Therefore, if $|J| < 2$ holds, an index set $J$ is advance shareable in our propsal for this stabilizer $S$.
\end{example}
\section{Conclusion}\label{sec:four}
In our paper, 
we propose a quantum secret sharing scheme that 
can distribute some shares 
before a given secret. 
In Section \ref{sec:3}, 
we give a construction of EAQECC from a stabilizer, 
whose erasure correcting ability is the same as the original stabilizer code.
Then, we clarify the access structures of our proposed quantum secret sharing. 
In Example 4, we confirm that our proposal can construct $(3,2,4)$ ramp QSS with 1-qubit shares.
This ramp QSS cannot be constructed by the schemes of Lie et al.\ \cite{Lie_2020} nor Ogawa et al.\ \cite{rampQSS_Ogawa_2005}.
We clarify a necessary and sufficient condition on advance shareable sets.
In Remark 1, our proposal shows the significance of constructing an EAQECC with a large number of maximally entangled pairs.
We give a sufficient condition of advance shareable set that can be verified without using the dimensions of linear spaces.
Therefore, our proposal can provide a useful method of advance sharing when a dealer unable to communicate with some participants after the dealer obtains a secret.

\section{Acknowledgments}
The author would like to thank Professor Tomohiko Uematsu for a helpful advice. 
This work was supported by JST SPRING, Grant Number JPMJSP2106.
\bibliography{papers_liblary}

\begin{thebibliography}{15}%
\makeatletter
\providecommand \@ifxundefined [1]{%
 \@ifx{#1\undefined}
}%
\providecommand \@ifnum [1]{%
 \ifnum #1\expandafter \@firstoftwo
 \else \expandafter \@secondoftwo
 \fi
}%
\providecommand \@ifx [1]{%
 \ifx #1\expandafter \@firstoftwo
 \else \expandafter \@secondoftwo
 \fi
}%
\providecommand \natexlab [1]{#1}%
\providecommand \enquote  [1]{``#1''}%
\providecommand \bibnamefont  [1]{#1}%
\providecommand \bibfnamefont [1]{#1}%
\providecommand \citenamefont [1]{#1}%
\providecommand \href@noop [0]{\@secondoftwo}%
\providecommand \href [0]{\begingroup \@sanitize@url \@href}%
\providecommand \@href[1]{\@@startlink{#1}\@@href}%
\providecommand \@@href[1]{\endgroup#1\@@endlink}%
\providecommand \@sanitize@url [0]{\catcode `\\12\catcode `\$12\catcode
  `\&12\catcode `\#12\catcode `\^12\catcode `\_12\catcode `\%12\relax}%
\providecommand \@@startlink[1]{}%
\providecommand \@@endlink[0]{}%
\providecommand \url  [0]{\begingroup\@sanitize@url \@url }%
\providecommand \@url [1]{\endgroup\@href {#1}{\urlprefix }}%
\providecommand \urlprefix  [0]{URL }%
\providecommand \Eprint [0]{\href }%
\providecommand \doibase [0]{https://doi.org/}%
\providecommand \selectlanguage [0]{\@gobble}%
\providecommand \bibinfo  [0]{\@secondoftwo}%
\providecommand \bibfield  [0]{\@secondoftwo}%
\providecommand \translation [1]{[#1]}%
\providecommand \BibitemOpen [0]{}%
\providecommand \bibitemStop [0]{}%
\providecommand \bibitemNoStop [0]{.\EOS\space}%
\providecommand \EOS [0]{\spacefactor3000\relax}%
\providecommand \BibitemShut  [1]{\csname bibitem#1\endcsname}%
\let\auto@bib@innerbib\@empty
\bibitem [{\citenamefont {Shamir}(1979)}]{SS_Shamir_1979}%
  \BibitemOpen
  \bibfield  {author} {\bibinfo {author} {\bibfnamefont {A.}~\bibnamefont
  {Shamir}},\ }\bibfield  {title} {\bibinfo {title} {How to share a secret},\
  }\href {https://doi.org/10.1145/359168.359176} {\bibfield  {journal}
  {\bibinfo  {journal} {Commun. ACM}\ }\textbf {\bibinfo {volume} {22}},\
  \bibinfo {pages} {612} (\bibinfo {year} {1979})}\BibitemShut {NoStop}%
\bibitem [{\citenamefont {Blakley}(1979)}]{SS_Blakley_1979}%
  \BibitemOpen
  \bibfield  {author} {\bibinfo {author} {\bibfnamefont {G.~R.}\ \bibnamefont
  {Blakley}},\ }\bibfield  {title} {\bibinfo {title} {Safeguarding
  cryptographic keys},\ }in\ \href@noop {} {\emph {\bibinfo {booktitle} {1979
  International Workshop on Manageing Requirements Knowledge (MARK)}}}\
  (\bibinfo {year} {1979})\ pp.\ \bibinfo {pages} {313--318}\BibitemShut
  {NoStop}%
\bibitem [{\citenamefont {Hillery}\ \emph {et~al.}(1999)\citenamefont
  {Hillery}, \citenamefont {Bu\ifmmode~\check{z}\else \v{z}\fi{}ek},\ and\
  \citenamefont {Berthiaume}}]{QSS_Hillery_1999}%
  \BibitemOpen
  \bibfield  {author} {\bibinfo {author} {\bibfnamefont {M.}~\bibnamefont
  {Hillery}}, \bibinfo {author} {\bibfnamefont {V.}~\bibnamefont
  {Bu\ifmmode~\check{z}\else \v{z}\fi{}ek}},\ and\ \bibinfo {author}
  {\bibfnamefont {A.}~\bibnamefont {Berthiaume}},\ }\bibfield  {title}
  {\bibinfo {title} {Quantum secret sharing},\ }\href
  {https://doi.org/10.1103/PhysRevA.59.1829} {\bibfield  {journal} {\bibinfo
  {journal} {Phys. Rev. A}\ }\textbf {\bibinfo {volume} {59}},\ \bibinfo
  {pages} {1829} (\bibinfo {year} {1999})}\BibitemShut {NoStop}%
\bibitem [{\citenamefont {Cleve}\ \emph {et~al.}(1999)\citenamefont {Cleve},
  \citenamefont {Gottesman},\ and\ \citenamefont {Lo}}]{QSS_Cleve_1999}%
  \BibitemOpen
  \bibfield  {author} {\bibinfo {author} {\bibfnamefont {R.}~\bibnamefont
  {Cleve}}, \bibinfo {author} {\bibfnamefont {D.}~\bibnamefont {Gottesman}},\
  and\ \bibinfo {author} {\bibfnamefont {H.-K.}\ \bibnamefont {Lo}},\
  }\bibfield  {title} {\bibinfo {title} {How to share a quantum secret},\
  }\href {https://doi.org/10.1103/PhysRevLett.83.648} {\bibfield  {journal}
  {\bibinfo  {journal} {Phys. Rev. Lett.}\ }\textbf {\bibinfo {volume} {83}},\
  \bibinfo {pages} {648} (\bibinfo {year} {1999})}\BibitemShut {NoStop}%
\bibitem [{\citenamefont {Gottesman}(2000)}]{Gottesman_QSS_2000}%
  \BibitemOpen
  \bibfield  {author} {\bibinfo {author} {\bibfnamefont {D.}~\bibnamefont
  {Gottesman}},\ }\bibfield  {title} {\bibinfo {title} {Theory of quantum
  secret sharing},\ }\href {http://dx.doi.org/10.1103/PhysRevA.61.042311}
  {\bibfield  {journal} {\bibinfo  {journal} {Phys. Rev. A}\ }\textbf {\bibinfo
  {volume} {61}} (\bibinfo {year} {2000})}\BibitemShut {NoStop}%
\bibitem [{\citenamefont {Senthoor}\ and\ \citenamefont
  {Sarvepalli}(2019)}]{efficientQSS_Senthoor_2021}%
  \BibitemOpen
  \bibfield  {author} {\bibinfo {author} {\bibfnamefont {K.}~\bibnamefont
  {Senthoor}}\ and\ \bibinfo {author} {\bibfnamefont {P.~K.}\ \bibnamefont
  {Sarvepalli}},\ }\bibfield  {title} {\bibinfo {title} {Communication
  efficient quantum secret sharing},\ }\href
  {https://doi.org/10.1103/PhysRevA.100.052313} {\bibfield  {journal} {\bibinfo
   {journal} {Phys. Rev. A}\ }\textbf {\bibinfo {volume} {100}},\ \bibinfo
  {pages} {052313} (\bibinfo {year} {2019})}\BibitemShut {NoStop}%
\bibitem [{\citenamefont {Senthoor}\ and\ \citenamefont
  {Sarvepalli}(2022{\natexlab{a}})}]{concatenatingQSS_Senthoor_2022}%
  \BibitemOpen
  \bibfield  {author} {\bibinfo {author} {\bibfnamefont {K.}~\bibnamefont
  {Senthoor}}\ and\ \bibinfo {author} {\bibfnamefont {P.~K.}\ \bibnamefont
  {Sarvepalli}},\ }\bibfield  {title} {\bibinfo {title} {Concatenating extended
  {CSS} codes for communication efficient quantum secret sharing},\ }\href@noop
  {} {\bibfield  {journal} {\bibinfo  {journal} {arXiv:2211.06910}\ } (\bibinfo
  {year} {2022}{\natexlab{a}})}\BibitemShut {NoStop}%
\bibitem [{\citenamefont {Senthoor}\ and\ \citenamefont
  {Sarvepalli}(2022{\natexlab{b}})}]{Senthoor_2022}%
  \BibitemOpen
  \bibfield  {author} {\bibinfo {author} {\bibfnamefont {K.}~\bibnamefont
  {Senthoor}}\ and\ \bibinfo {author} {\bibfnamefont {P.~K.}\ \bibnamefont
  {Sarvepalli}},\ }\bibfield  {title} {\bibinfo {title} {Theory of
  communication efficient quantum secret sharing},\ }\href
  {https://doi.org/10.1109/TIT.2021.3139839} {\bibfield  {journal} {\bibinfo
  {journal} {IEEE Transactions on Information Theory}\ }\textbf {\bibinfo
  {volume} {68}},\ \bibinfo {pages} {3164} (\bibinfo {year}
  {2022}{\natexlab{b}})}\BibitemShut {NoStop}%
\bibitem [{\citenamefont {Ogawa}\ \emph {et~al.}(2005)\citenamefont {Ogawa},
  \citenamefont {Sasaki}, \citenamefont {Iwamoto},\ and\ \citenamefont
  {Yamamoto}}]{rampQSS_Ogawa_2005}%
  \BibitemOpen
  \bibfield  {author} {\bibinfo {author} {\bibfnamefont {T.}~\bibnamefont
  {Ogawa}}, \bibinfo {author} {\bibfnamefont {A.}~\bibnamefont {Sasaki}},
  \bibinfo {author} {\bibfnamefont {M.}~\bibnamefont {Iwamoto}},\ and\ \bibinfo
  {author} {\bibfnamefont {H.}~\bibnamefont {Yamamoto}},\ }\bibfield  {title}
  {\bibinfo {title} {Quantum secret sharing schemes and reversibility of
  quantum operations},\ }\bibfield  {journal} {\bibinfo  {journal} {Phys. Rev.
  A}\ }\textbf {\bibinfo {volume} {72}},\ \href {https://doi.org/032318}
  {032318} (\bibinfo {year} {2005})\BibitemShut {NoStop}%
\bibitem [{\citenamefont {Miyajima}\ and\ \citenamefont
  {Matsumoto}(2022)}]{Miyajima_2022}%
  \BibitemOpen
  \bibfield  {author} {\bibinfo {author} {\bibfnamefont {R.}~\bibnamefont
  {Miyajima}}\ and\ \bibinfo {author} {\bibfnamefont {R.}~\bibnamefont
  {Matsumoto}},\ }\bibfield  {title} {\bibinfo {title} {Advance sharing of
  quantum shares for classical secrets},\ }\href
  {https://doi.org/10.1109/access.2022.3204389} {\bibfield  {journal} {\bibinfo
   {journal} {{IEEE} Access}\ }\textbf {\bibinfo {volume} {10}},\ \bibinfo
  {pages} {94458} (\bibinfo {year} {2022})}\BibitemShut {NoStop}%
\bibitem [{\citenamefont {Lie}\ and\ \citenamefont {Jeong}(2020)}]{Lie_2020}%
  \BibitemOpen
  \bibfield  {author} {\bibinfo {author} {\bibfnamefont {S.~H.}\ \bibnamefont
  {Lie}}\ and\ \bibinfo {author} {\bibfnamefont {H.}~\bibnamefont {Jeong}},\
  }\bibfield  {title} {\bibinfo {title} {Randomness cost of masking quantum
  information and the information conservation law},\ }\href
  {https://doi.org/10.1103/PhysRevA.101.052322} {\bibfield  {journal} {\bibinfo
   {journal} {Phys. Rev. A}\ }\textbf {\bibinfo {volume} {101}},\ \bibinfo
  {pages} {052322} (\bibinfo {year} {2020})}\BibitemShut {NoStop}%
\bibitem [{\citenamefont {Brun}\ \emph {et~al.}(2006)\citenamefont {Brun},
  \citenamefont {Devetak},\ and\ \citenamefont {Hsieh}}]{EAQECC_Brun_2006}%
  \BibitemOpen
  \bibfield  {author} {\bibinfo {author} {\bibfnamefont {T.}~\bibnamefont
  {Brun}}, \bibinfo {author} {\bibfnamefont {I.}~\bibnamefont {Devetak}},\ and\
  \bibinfo {author} {\bibfnamefont {M.-H.}\ \bibnamefont {Hsieh}},\ }\bibfield
  {title} {\bibinfo {title} {Correcting quantum errors with entanglement},\
  }\href {https://doi.org/10.1126/science.1131563} {\bibfield  {journal}
  {\bibinfo  {journal} {Science}\ }\textbf {\bibinfo {volume} {314}},\ \bibinfo
  {pages} {436} (\bibinfo {year} {2006})}\BibitemShut {NoStop}%
\bibitem [{\citenamefont {Luo}\ \emph {et~al.}(2017)\citenamefont {Luo},
  \citenamefont {Ma}, \citenamefont {Wei},\ and\ \citenamefont
  {Leng}}]{NonbinaryEAQECC_Luo_2017}%
  \BibitemOpen
  \bibfield  {author} {\bibinfo {author} {\bibfnamefont {L.}~\bibnamefont
  {Luo}}, \bibinfo {author} {\bibfnamefont {Z.}~\bibnamefont {Ma}}, \bibinfo
  {author} {\bibfnamefont {Z.}~\bibnamefont {Wei}},\ and\ \bibinfo {author}
  {\bibfnamefont {R.}~\bibnamefont {Leng}},\ }\bibfield  {title} {\bibinfo
  {title} {Non-binary entanglement-assisted quantum stabilizer codes},\
  }\href@noop {} {\bibfield  {journal} {\bibinfo  {journal} {Sci. China Inf.
  Sci.}\ }\textbf {\bibinfo {volume} {60}},\ \bibinfo {pages} {042501}
  (\bibinfo {year} {2017})}\BibitemShut {NoStop}%
\bibitem [{\citenamefont
  {Matsumoto}(2017)}]{Unitary_reconstruction_Matsumoto_2017}%
  \BibitemOpen
  \bibfield  {author} {\bibinfo {author} {\bibfnamefont {R.}~\bibnamefont
  {Matsumoto}},\ }\bibfield  {title} {\bibinfo {title} {Unitary reconstruction
  of secret for stabilizer-based quantum secret sharing},\ }\href
  {https://doi.org/10.1007/s11128-017-1656-1} {\bibfield  {journal} {\bibinfo
  {journal} {Quantum Information Processing}\ }\textbf {\bibinfo {volume}
  {16}},\ \bibinfo {pages} {202} (\bibinfo {year} {2017})}\BibitemShut
  {NoStop}%
\bibitem [{\citenamefont {Ueno}\ and\ \citenamefont
  {Matsumoto}(2022)}]{Ueno_2022}%
  \BibitemOpen
  \bibfield  {author} {\bibinfo {author} {\bibfnamefont {D.}~\bibnamefont
  {Ueno}}\ and\ \bibinfo {author} {\bibfnamefont {R.}~\bibnamefont
  {Matsumoto}},\ }\bibfield  {title} {\bibinfo {title} {Explicit method to make
  shortened stabilizer {EAQECC} from stabilizer {QECC}},\ }\href@noop {}
  {\bibfield  {journal} {\bibinfo  {journal} {arXiv:2205.13732}\ } (\bibinfo
  {year} {2022})}\BibitemShut {NoStop}%
\end{thebibliography}%


\end{document}